\documentclass[letterpaper, 10 pt, conference]{ieeeconf}  %
\IEEEoverridecommandlockouts     

\usepackage{balance} 
\usepackage{xcolor,calc}

\usepackage[pdftex]{graphicx}
\usepackage{amsmath}         
\usepackage{amssymb}         
\usepackage[noadjust]{cite}
\usepackage{color}
\usepackage{enumerate}
\usepackage{hyperref}
\usepackage{multirow}
\usepackage{mathtools}
\usepackage{subfig}
\usepackage{booktabs}
\usepackage{epstopdf}
\usepackage{algcompatible}
\usepackage[linesnumbered,ruled,vlined]{algorithm2e}

\usepackage{etoolbox}
\makeatletter
\patchcmd{\@makecaption}
  {\scshape}
  {}
  {}
  {}
\makeatother

\usepackage[normalem]{ulem}

\usepackage{url}
\usepackage{bm}

\usepackage{booktabs}

\newtheorem{theorem}{Theorem}
\newtheorem{remark}{Remark}

\SetKwInput{KwInit}{Initialization}                
\SetKwInput{KwInput}{Input}
\SetKwInput{KwOutput}{Output}
\makeatletter
\newsavebox\myboxA
\newsavebox\myboxB
\newlength\mylenA

\newcommand*\xoverline[2][0.75]{%
    \sbox{\myboxA}{$\m@th#2$}%
    \setbox\myboxB\null
    \ht\myboxB=\ht\myboxA%
    \dp\myboxB=\dp\myboxA%
    \wd\myboxB=#1\wd\myboxA
    \sbox\myboxB{$\m@th\overline{\copy\myboxB}$}
    \setlength\mylenA{\the\wd\myboxA}
    \addtolength\mylenA{-\the\wd\myboxB}%
    \ifdim\wd\myboxB<\wd\myboxA%
       \rlap{\hskip 0.5\mylenA\usebox\myboxB}{\usebox\myboxA}%
    \else
        \hskip -0.5\mylenA\rlap{\usebox\myboxA}{\hskip 0.5\mylenA\usebox\myboxB}%
    \fi}
\makeatother

\usepackage{tabularx}
\usepackage{booktabs}

\newtheorem{definition}{Definition}
\newtheorem{assumption}{Assumption}

\makeatletter
\let\NAT@parse\undefined
\makeatother

\title{\LARGE \bf A Simple Robust MPC for Linear Systems with \\ Parametric and Additive Uncertainty}

\author{Monimoy Bujarbaruah$^1$, Ugo Rosolia$^2$, Yvonne R. St{\"u}rz$^1$, Francesco Borrelli$^1$  
\thanks{$^1$MPC Lab, UC Berkeley, CA 94720, USA; E-mails: \tt\scriptsize{\{monimoyb, y.stuerz, fborrelli\}@berkeley.edu.}} \thanks{$^2$Caltech, CA 91125, USA; E-mail: \tt\scriptsize{urosolia@caltech.edu.}}
}

\begin{document}

\maketitle
  \thispagestyle{empty}
\pagestyle{empty}

\begin{abstract}
We propose a simple and computationally efficient approach for designing a robust Model Predictive Controller (MPC) for constrained uncertain linear systems. The uncertainty is modeled as an additive disturbance and an additive error on the system dynamics matrices. Set based bounds for each component of the model uncertainty are assumed to be known. We separate the constraint tightening strategy into two parts, depending on the length of the MPC horizon. For a horizon length of one, the robust MPC problem is solved exactly, whereas for other horizon lengths, the model uncertainty is over-approximated with a net-additive component. The resulting MPC controller guarantees robust satisfaction of state and input constraints in closed-loop with the uncertain system. With appropriately designed terminal components and an adaptive horizon strategy, we prove the controller's recursive feasibility and stability of the origin. With numerical simulations, we demonstrate that our proposed approach gains up to 15x online computation speedup over a tube MPC strategy, while stabilizing about 98$\%$ of the latter's region of attraction.   
\end{abstract}


\section{Introduction}
Model Predictive Control (MPC) is an 
optimal control strategy that satisfies imposed constraints on system states and inputs \cite{mayne2000constrained,kouvaritakis2016model, borrelli2017predictive}. 
The presence of uncertainty in the prediction model is a key challenge in MPC design.  
For uncertain linear systems in presence of \emph{only} an additive disturbance in the system model, as finding the optimal policy is NP-hard, 
computationally tractable suboptimal robust MPC techniques such as tube MPC \cite{chisci2001systems, langson2004robust, Goulart2006, rakovic2012homothetic, kouvaritakis2016model, bujarbaruahAdapCDC18} have been widely utilized. The idea in these techniques is to restrict the input policy to the space of affine state feedback policies and then tightening the imposed constraints around a predicted nominal (i.e., certainty-equivalent) trajectory within a tube. This ensures that the realized system trajectory satisfies imposed constraints robustly for all disturbances in the system. 

Robust MPC design for uncertain linear systems in presence of \emph{both} a mismatch in the system dynamics matrices and an additive disturbance is computationally more intensive and is a topic of ongoing research \cite{rakovic2013homothetic, munoz2013recursively, slsmpc}. In order to design a computationally efficient classical shrinking or fixed radius tube MPC \cite[Chapter~3]{kouvaritakis2016model} 
in presence of mismatch in the system matrices, the contribution of uncertainty due to the mismatches can be lumped together with the additive disturbance. A worst-case bound for this quantity can be found and then a method such as \cite{Goulart2006} can be used. However, the work in \cite{dean2018safely} points out that such an approach can lead to overly conservative behavior, which they 
circumvent by utilizing a System Level Synthesis based approach. 

In this paper we show that such a naive ``net-additive" uncertainty approach may not always lead to overly conservative behavior over \cite{dean2018safely}, if the terminal constraints are appropriately chosen and an adaptive horizon strategy is adopted. Our method can also be used to obtain a single roll-out policy for robust constraint satisfaction, without solving the MPC problem repeatedly. 
Our key contributions are:
\begin{itemize}
    \item We split the constraint tightening into two cases based on the horizon length. For horizon of one, the robust MPC problem is solved exactly. For larger horizons, we lump the model uncertainty into a net-additive component and compute constraint tightenings along the prediction horizon based on its worst-case bound. 
    
    \item We solve a set of tractable convex optimization problems online using an adaptive horizon approach for the MPC controller synthesis. With an appropriately constructed terminal set and a terminal cost we prove recursive feasibility of the controller synthesis problem in closed-loop and input to state stability of the origin.
    
    \item We numerically compare our proposed robust MPC approach 
    with the tube MPC from \cite{langson2004robust} and the constrained LQR algorithm from \cite{dean2018safely}. 
    In the first case, we gain up to 15x speedup of online control computations while stabilizing approximately 98$\%$ of the tube MPC's region of attraction (ROA). In the latter case, our approach obtains an up to 12x larger ROA with the open-loop roll-out policy.  
\end{itemize}
\subsection*{Notation}
The induced $p$-norm of any matrix $A$ is given by $\Vert A\Vert_p = \sup_{x \neq 0} \frac{\Vert Ax\Vert_p}{\Vert x\Vert_p}$, where $\Vert \cdot \Vert_p$ is the $p$-norm of a vector. 
The sign $u \geq v$ between two vectors $u,v$ denotes element-wise inequality. 
The convex combination of the matrices $X, Y$ is denoted as $\mathrm{conv}(X, Y)$. $\mathcal{A} \oplus \mathcal{B}$ denotes the Minkowski sum of the two sets $\mathcal{A}$ and $\mathcal{B}$. $A \otimes B$ denotes Kronecker product. 
$I_n$ denotes an identity matrix of size $n$. 
Consistency property is $\Vert Xy\Vert_q \leq \Vert X\Vert_p \Vert y\Vert_q$, for any matrix $X$ and vector $y$.

\section{Problem Formulation}\label{sec:prob}
We consider the linear system 
\begin{equation}\label{eq:unc_system}
     x_{t+1} = A x_t + B u_t + w_t,~x_0 = x_S,
\end{equation}
where $x_t\in \mathbb{R}^{d}$ is the state and $u_t\in\mathbb{R}^{m}$ is the input at time step $t$, and $A$ and $B$ are system dynamics matrices of appropriate dimensions. We assume that $A$ and $B$ are unknown matrices with estimates $\bar{A}$ and $\bar{B}$ available to the control designer. In particular, we consider
\begin{align}\label{eq:matrix_errors}
    & A = \bar{A} + \Delta^\mathrm{tr}_A,~B = \bar{B} + \Delta^\mathrm{tr}_B,
\end{align}
where the true parametric uncertainty matrices $\Delta^\mathrm{tr}_A$ and $\Delta^\mathrm{tr}_B$ are unknown and belong to convex and compact sets \begin{align}\label{err_in_sets_pol}
    & \Delta^\mathrm{tr}_A \in \mathcal{P}_A,~\Delta^\mathrm{tr}_B \in \mathcal{P}_B.
\end{align}
Furthermore, we consider that the sets $\mathcal{P}_A$ and $\mathcal{P}_B$ are convex hulls of known \emph{vertex} matrices $\{\Delta_A^{(1)}, \Delta_A^{(2)}, \dots, \Delta_A^{(n_a)}\}$ and $\{\Delta_B^{(1)}, \Delta_B^{(2)}, \dots, \Delta_B^{(n_b)}\}$, with fixed $n_a, n_b >0$:
\begin{subequations}\label{eq:pol_out_termset_cond}
\begin{align}
    & \mathcal{P}_A = \mathrm{conv}(\Delta_A^{(1)}, \Delta_A^{(2)}, \dots, \Delta_A^{(n_a)}),\\
    & \mathcal{P}_B = \mathrm{conv}(\Delta_B^{(1)}, \Delta_B^{(2)}, \dots, \Delta_B^{(n_b)}).
\end{align}
\end{subequations}
System~\eqref{eq:unc_system} is also affected by a disturbance $w_t$ with a convex and compact support $\mathbb{W} \subset \mathbb{R}^{d}$, at all times $t \geq 0$.

We are interested in synthesizing a robust MPC for \eqref{eq:unc_system}, by repeatedly solving the following optimal control problem:
\begin{subequations}\label{eq:generalized_InfOCP}
	\begin{align}
		\displaystyle\min_{U_t(\cdot)} ~ & \displaystyle\sum\limits_{k = t}^{t+N-1} \ell \left( \bar{x}_{k|t}, u_{k|t}\left(\bar{x}_{k|t}\right) \right)+Q(\bar{x}_{t+N|t}) \label{eq:FTOCP_cost} \\
		\text{s.t.,} & ~~~ \bar{x}_{k+1|t} = \bar{A} \bar{x}_{k|t} + \bar{B} u_{k|t}(\bar{x}_{k|t}), \label{eq:FTOCP_nominal} \\
		& ~~~x_{k+1|t} = Ax_{k|t} + Bu_{k|t}(x_{k|t}) + w_{k|t}, \label{eq:FTOCP_trueModel} \\
		&~~~\textnormal{with } A = \bar{A} + \Delta_A, B = \bar{B} + \Delta_B, \label{eq:FTOCP_modelUncertanty}\\[1ex]
		&~~~ H^x x_{k|t} \leq h^x,
		H^u u_{k|t}(x_{k|t}) \leq h^u, \label{eq:FTOCP_constr} \\
		&~~~x_{t+N|t} \in \mathcal{X}_N, \label{FTOCP_termC}\\[1ex]
		&~~~\forall w_{k|t} \in \mathbb W,~\forall \Delta_A \in \mathcal{P}_A,~\forall \Delta_B \in \mathcal{P}_B, \label{eq:FTOCP_uncertanty} \\
		&~~~\forall k \in \{t,t+1,\dots,(t+N-1)\}, \nonumber\\
		&~~~x_{t|t}=\bar{x}_{t|t} = x_t \nonumber,
	\end{align}
\end{subequations}
with $U_t(\cdot) = \{u_{t|t},u_{t+1|t}(\cdot),\dots,u_{t+N-1|t}(\cdot)\}$, and applying the optimal MPC policy
\begin{align}\label{eq:mpc_pol_formulation}
    u^\mathrm{MPC}_t(x_t) = u^\star_{t|t}(x_t),
\end{align}
to system \eqref{eq:unc_system} in closed-loop, where $x_{k|t}$ is the predicted state at time step $k$ for any possible uncertainty realization, obtained by applying the predicted input policies $\{u_{t|t},u_{t+1|t}(\cdot),\dots,u_{k-1|t}(\cdot)\}$ to system~\eqref{eq:unc_system}, and $\{\bar{x}_{k|t}, \bar{u}_{k|t}\}$ with $\bar{u}_{k|t} = u_{k|t}(\bar{x}_{k|t})$ denote the nominal state and corresponding input respectively. 
The constraints~\eqref{eq:FTOCP_constr}-\eqref{FTOCP_termC} are satisfied for all uncertainty realizations in~\eqref{eq:FTOCP_uncertanty}, where $H^x \in \mathbb{R}^{s \times d}, h^x \in \mathbb{R}^s, H^u \in \mathbb{R}^{o \times m}$ and $h^u \in \mathbb{R}^o$ parametrize \emph{compact} sets. Finally, the stage cost $\ell( x, u) = x^\top P x + u^\top R u$, and the terminal cost $Q(x) = x^\top P_N x$. 
The main challenges with solving \eqref{eq:generalized_InfOCP} are:
\begin{enumerate}[(A)]
    \item  The state and input constraints are to be satisfied robustly under the presence of mismatch in the system dynamics matrices and disturbances. That is, \eqref{eq:FTOCP_constr}-\eqref{FTOCP_termC}-\eqref{eq:FTOCP_uncertanty} need to be reformulated for a numerical algorithm.
    \item Optimizing over policies $\{u_0, u_1(\cdot), u_2(\cdot), \dots\}$ in \eqref{eq:generalized_InfOCP} 
    is not tractable in general for constrained linear systems.
    \item The feasibility of problem \eqref{eq:generalized_InfOCP} is to be guaranteed robustly at all time steps $t \geq 0$. 
    That is, 
    \begin{align*}
    &H^x x_t \leq h^x, H^u u^\mathrm{MPC}_t(x_t) \leq h^u,~\forall w_t \in \mathbb{W},\forall t \geq 0,
    \end{align*}
where $x_{t+1} = Ax_t + Bu^\mathrm{MPC}_t(x_t) + w_t$.
\end{enumerate}
Methods addressing Challenge (B) and Challenge (C) are well established in MPC literature. 
In the following sections, we show how we address Challenge (A). 
\vspace{3pt}\\
\noindent \textbf{Approach Insight}:  We lump the component of model mismatch together with the additive disturbance into a ``net-additive" uncertainty. 
We design a simple and computationally efficient shrinking tube MPC leveraging worst-case bounds of this net-additive uncertainty \emph{only} along the prediction horizon, and an \emph{exact} system uncertainty representation for the construction of the terminal set.
Notice that shrinking tube MPC strategies such as \cite{Goulart2006,chisci2001systems} using the net-additive uncertainty bounds \textit{both} along the prediction horizon and for the computation of the terminal set, can be extremely conservative as pointed out in \cite{rawlings2009model, dean2018safely}. Therefore, numerous alternative strategies such as polytopic, homothetic and elastic tube MPC \cite{langson2004robust,munoz2013recursively,rakovic2013homothetic} have been introduced to lower this conservatism by circumventing the net-additive uncertainty bounds. This however increases the online computation times of these algorithms \cite{kouvaritakis2016model,slsmpc}.
In Section~\ref{sec:numerics}, we show with numerical simulations that our approach balances the trade-off between conservatism and computational complexity. In the example under study, we obtain about 15x online computation speedup over the polytopic tube MPC method of \cite{langson2004robust}, while stabilizing about 98$\%$ of its region of attraction.
\section{Robust MPC Design}\label{sec:mpc}
In this section we present the steps of the proposed robust MPC design approach, which solves problem \eqref{eq:generalized_InfOCP} at all $t \geq 0$. 
\subsection{Net-Additive Uncertainty Representation}\label{ssec:net_add}
We lump the effect of the parametric uncertainty and the additive disturbance into an augmented disturbance $\tilde{w}_t$. We denote $\tilde{w}_t = \Delta_A x_t + \Delta_B u_t + w_t$, and $\Vert \tilde{w}_t\Vert \leq \tilde{w}_\mathrm{max}$, for all $t \geq 0$, with the bound $\tilde{w}_\mathrm{max} = \max_{t \geq 0}\Vert \tilde{w}_t \Vert$ computed as:
\begin{subequations}\label{bounds_eq}
\begin{align}
   \max_{t \geq 0} \Vert \tilde{w}_t \Vert & \leq \max_{t \geq 0} (\Vert \Delta_A x_t \Vert + \Vert \Delta_B u_t \Vert + \Vert w_t \Vert), \label{dboundTr}\\
    & \!\!\! \! \!\!\! \!\! \! \!\! \leq \max_{t \geq 0} (\Vert \Delta_A \Vert_p \Vert x_t \Vert + \Vert \Delta_B \Vert_p \Vert u_t \Vert + \Vert w_t \Vert), \label{dboundCon}\\
    & \!\!\! \! \!\! \! \!\! \! \!\! = \Vert \Delta_A \Vert_p \Vert x \Vert_\mathrm{max} + \Vert \Delta_B \Vert_p \Vert u \Vert_\mathrm{max} + \Vert w \Vert_\mathrm{max}, \label{dboundMax}\\
    & \!\!\! \! \!\! \! \!\! \! \!\! = \tilde{w}_\mathrm{max}, \nonumber
\end{align}
\end{subequations}
with $\Delta_A \in \mathcal{P}_A$ and $\Delta_B \in \mathcal{P}_B$. In \eqref{dboundTr} we have used the triangle inequality and in \eqref{dboundCon} the consistency property of induced norms. Values of $\Vert x \Vert_\mathrm{max}$, $\Vert u \Vert_\mathrm{max}$ and $\Vert w \Vert_\mathrm{max}$ in \eqref{dboundMax} can be obtained from compact constraints \eqref{eq:FTOCP_constr} and $\mathbb{W}$.

\subsection{Control Policy Parametrization}\label{sec:pol_par}
To address Challenge (B), for all predicted steps $k \in \{t,t+1,\dots,t+N-1\}$ over the MPC horizon, the control policy is chosen as \cite{lofberg2003minimax,Goulart2006}:
\begin{equation}\label{eq:inputParam_DF_OL}
	u_{k|t}(x_{k|t}) = \sum \limits_{l=t}^{k-1}M_{k,l|t} \tilde{w}_{l|t}  + \bar{u}_{k|t},
\end{equation}
where $M_{k|t}$ are the \emph{planned} feedback gains at time step $t$ and $\bar{u}_{k|t} = u_{k|t}(\bar{x}_{k|t})$ are the nominal inputs. Then the sequence of predicted inputs can be written as $\mathbf{u}_t = \mathbf{M}^{(N)}_t \tilde{\mathbf{w}}_t + \mathbf{\bar{u}}^{(N)}_t$, where $\mathbf{M}^{(N)}_t\in \mathbb{R}^{mN \times dN}$ and $\mathbf{\bar{u}}^{(N)}_t \in \mathbb{R}^{mN}$ are
\begin{equation*}
\begin{aligned}
 &   \mathbf{M}^{(N)}_t  =  \begin{bmatrix}0& \dots&\dots&0\\
  M_{t+1,t}& 0 & \dots & 0\\
  \vdots &\ddots& \ddots &\vdots\\
  M_{t+N-1,t}& \dots& M_{t+N-1,t+N-2}& 0
  \end{bmatrix},\\
 &  \mathbf{\bar{u}}^{(N)}_t = [\bar{u}_{t|t}^\top,\bar{u}_{t+1|t}^\top , \dots, \bar{u}_{t+N-1|t}^\top]^\top,
\end{aligned}
\end{equation*}
and 
\begin{equation*}
\begin{aligned}
& \mathbf{u}_t = [{u}^\top_{t|t}, {u}^\top_{t+1|t}(\cdot), \dots, {u}^\top_{t+N-1|t}(\cdot)]^\top,\\ 
& \tilde{\mathbf{w}}_t = \begin{bmatrix}
     \tilde{w}_{t|t}^\top &
     \tilde{w}_{t+1|t}^\top &
    \dots 
     \tilde{w}_{t+N-1|t}^\top
    \end{bmatrix}^\top,
\end{aligned}
\end{equation*}
with $\Vert \tilde{\mathbf{w}}_t \Vert \leq \tilde{\mathbf{w}}_\mathrm{max}$ for all $t \geq 0$. 
\subsection{Terminal Set Construction}\label{ssec:term_set}
We present the construction of the terminal set $\mathcal{X}_N$ in this section to address Challenge (C) mentioned in Section~\ref{sec:mpc}. Consider a linear state feedback policy for constructing $\mathcal{X}_N$
\begin{align}\label{eq:term_pol}
    \kappa_N(x) = Kx,
\end{align}
where $K \in \mathbb{R}^{m \times d}$ is the feedback gain. Recall the sets $\mathcal{P}_A$ and $\mathcal{P}_B$ from \eqref{eq:pol_out_termset_cond}. We define 
\begin{align*}
    &\mathcal{P}_{A_\Delta} = \{A_m: A_m = \bar{A} + \Delta_A,~\forall \Delta_A \in \mathcal{P}_A\},\\
    & \mathcal{P}_{B_\Delta} = \{B_m: B_m = \bar{B} + \Delta_B,~\forall \Delta_B \in \mathcal{P}_B\}. 
\end{align*}
Under policy \eqref{eq:term_pol}, the  closed-loop system dynamics matrix considered for constructing the terminal set satisfies
\begin{align*}
A^\mathrm{cl}=A + BK \in \mathcal{P}_{A_\Delta} \oplus \mathcal{P}_{B_\Delta}K. 
\end{align*}

\begin{assumption}\label{assump:stable}
$A^\mathrm{cl}_m = (A_m + B_mK)$ is stable for all $A_m \in \mathcal{P}_{A_\Delta}$ and $B_m \in \mathcal{P}_{B_\Delta}$. 
\end{assumption}

Using Assumption~\ref{assump:stable}, the terminal set ${\mathcal{X}}_N$ can then be computed as the maximal robust positive invariant set for 
\begin{align*}
    x_{t+1} = (A_m+B_mK) x_t + w_t,
\end{align*}
for all $A_m \in \mathcal{P}_{A_\Delta}, B_m \in \mathcal{P}_{B_\Delta}$, and for all $w_t \in \mathbb{W}$. That is for all $x\in \mathcal{X}_N$ we have that
\begin{align}\label{eq:term_set_DF}
   & H_x x \leq h_x,~H_u Kx \leq h_u\text{ and }(A_m + B_mK)x + w \in \mathcal{X}_N, \nonumber \\
    &\forall A_m \in \mathcal{P}_{A_\Delta},~\forall B_m \in \mathcal{P}_{B_\Delta},~\forall w \in \mathbb{W}.
\end{align}
\subsection{MPC Problem with Adaptive Horizon}\label{ssec:mpc_problem}
We now present the MPC reformulation of \eqref{eq:generalized_InfOCP} which guarantees recursive feasibility and Input to State Stability.
Note, the terminal set $\mathcal{X}_N$ is robustly invariant to all uncertainty of the form: $\forall \Delta_A \in \mathcal{P}_A,~\forall \Delta_B \in \mathcal{P}_B,~\forall w \in \mathbb{W},~\forall t \geq 0$, when the state feedback policy $\kappa_N(x) = Kx$ is used in closed-loop with system~\eqref{eq:unc_system}. However, along the prediction horizon we synthesize bound \eqref{bounds_eq} using more conservative tightenings from H{\"o}lder's and triangle inequalities, and the induced norm consistency property. Thus the uncertainty bounds along the horizon over-approximate the effect of the true uncertainty used to compute the  terminal set. This implies that the classical shifting argument \cite[Chapter~12]{borrelli2017predictive} for recursive MPC feasibility cannot be used. To resolve this issue, we solve a set of $N$ convex optimization problems at any $t$ for control synthesis, with the prediction horizon
$N_t \in \{1,2,\dots,N\}$. 
If one of these $N$ problems is feasible at time step $0$, we guarantee feasibility of at least one of them for all $t \geq 0$.

We first use policy \eqref{eq:inputParam_DF_OL} to reformulate the robust state constraints in \eqref{eq:generalized_InfOCP} along and at the end of the prediction horizon. Let the terminal set $\mathcal{X}_N$ in \eqref{eq:term_set_DF} be defined by 
$\mathcal{X}_N = \{x: H^x_N x \leq h^x_N\}$, with $H^x_N \in \mathbb{R}^{r \times d}, h^x_N \in \mathbb{R}^{r}$. For a horizon length of $N_t$, we denote matrices $\mathbf{F}^x = \mathrm{diag}(I_{N_t-1} \otimes H^x, H^x_N) \in \mathbb{R}^{(s(N_t-1)+r) \times dN_t}$ and $\mathbf{f}^x = [(h^x)^\top, (h^x)^\top, \dots, (h_N^x)^\top ]^\top \in \mathbb{R}^{s(N_t-1)+r}$. Also denote the set $\tilde{\mathbf{W}} = \{\tilde{\mathbf{w}} \in \mathbb{R}^{dN_t}: \Vert \tilde{\mathbf{w}}_t \Vert \leq \tilde{\mathbf{w}}_\mathrm{max}\}$. Then we consider the following two cases as\footnote{Note that the dimensions of $\mathbf{F}^x$, $\mathbf{f}^x$, $\bar{\mathbf{A}}$, $\mathbf{C}$, $\mathbf{G}$ and $\tilde{\mathbf{w}}_t$ vary depending on $N_t$. We omit showing this dependence explicitly for brevity.}:
\begin{subequations}\label{eq:state_robcon1}
\begin{align}
& \textnormal{\textbf{Case 1: $N_t = 1$:}} \nonumber \\ 
& \max_{\substack{{w}_t \in \mathbb{W}\\ \Delta_A \in \mathcal{P}_A \\ \Delta_B \in \mathcal{P}_B}}  H_N^x(\bar{{A}}+{\Delta}_A) x_t + (\bar{{B}} + {\Delta}_B) \bar{\mathbf{u}}^{(1)}_t + {w}_t) \leq h_N^x, \label{n1} \\
& \textnormal{\textbf{Case 2: $N_t \geq 2$:}} \nonumber \\
& \max_{\tilde{\mathbf{w}}_t \in \tilde{\mathbf{W}}}  \mathbf{F}^x \Big (  \bar{\mathbf{A}} x_t + \mathbf{C} \bar{\mathbf{u}}^{(N_t)}_t + (\mathbf{C}\mathbf{M}_t^{(N_t)} + \mathbf{G}) \tilde{\mathbf{w}}_t \Big ) \leq \mathbf{f}^x, \label{ng1}
\end{align}
\end{subequations}
where matrices $\bar{\mathbf{A}}, \mathbf{C}$ and $\mathbf{G}$ are defined in the Appendix. 
\begin{remark}
In \eqref{n1} we exactly propagate the system uncertainty for robustification. This ensures the feasibility of \eqref{n1} inside $\mathcal{X}_N$, which is a robust positive invariant set computed from \eqref{eq:term_set_DF} also using the exact uncertainty representation. As such uncertainty propagation is computationally intense over multi step predictions, in \eqref{ng1} we over-approximate system uncertainty using bounds \eqref{bounds_eq}.  
\end{remark}

Now, denote the matrices $\mathbf{H}^u = I_{N_t} \otimes H^u \in \mathbb{R}^{oN_t \times m N_t}$, and $\mathbf{h}^u = [({h}^u)^\top, ({h}^u)^\top, \dots, ({h}^u)^\top]^\top \in \mathbb{R}^{o N_t}$. Once the state constraints are formulated, the input constraints in \eqref{eq:generalized_InfOCP} along the prediction horizon can be written as:
\begin{align}\label{eq:input_robcon} 
    & \max_{\tilde{\mathbf{w}}_t \in \tilde{\mathbf{W}}}  \mathbf{H}^u \Big ( \mathbf{M}^{(N_t)}_t \tilde{\mathbf{w}}_t + \bar{\mathbf{u}}^{(N_t)}_t \Big) \leq \mathbf{h}^u, 
\end{align}
for $N_t \in \{1,2,\dots, N\}$. Using \eqref{eq:state_robcon1}-\eqref{eq:input_robcon},  we solve at any $t$: 
\begin{equation}\label{eq:MPC_R_fin_trac}
	\begin{aligned} 
	  &V_{t \rightarrow t+N_t}^{\mathrm{MPC}}(x_t, N_t)   :=  \\
	& \min_{\mathbf{M}^{(N_t)}_t, \bar{\mathbf{u}}^{(N_t)}_t} ~ 
	\begin{bmatrix}(\bar{\mathbf{x}}^{(N_t)}_t)^\top & (\bar{\mathbf{u}}^{(N_t)}_t)^\top \end{bmatrix} \bar Q^{(N_t)} \begin{bmatrix} \bar{\mathbf{x}}^{(N_t)}_t \\ \bar{\mathbf{u}}^{(N_t)}_t\end{bmatrix} \\
	& ~~~~~~\text{s.t., }~~~~~ \bar{\mathbf{x}}^{(N_t)}_t = \bar{\mathbf{A}} x_t + \mathbf{C} \bar{\mathbf{u}}^{(N_t)}_t, \\ 
       & ~~~~~~~~~~~~~~~~ \eqref{n1},\eqref{eq:input_robcon}~\textnormal{if $N_t = 1$, else } \eqref{ng1},\eqref{eq:input_robcon},\\
        &~~~~~~~~~~~~~~~~ \forall k = \{t,t+1,\dots,t+N_t-1\},  \\
        &~~~~~~~~~~~~~~~~\bar{x}_{t|t} = x_t,
	\end{aligned}
\end{equation}
for $N_t \in \{1,2,\dots, N\}$, where $\bar Q^{(N_t)} = \text{diag}(I_{N_t} \otimes P, P_N, I_{N_t} \otimes R)$. We reformulate \eqref{eq:MPC_R_fin_trac} as a convex program with standard duality arguments.  After solving \eqref{eq:MPC_R_fin_trac} for $N_t\in\{1,2,\dots, N\}$, we set
\begin{equation}\label{eq:nstar}
    N^\star_t = \arg \min_{\bar N \in \{1,2,\dots, N\}} V_{t \rightarrow t+N_t}^{\mathrm{MPC}}(x_t, \bar N).
\end{equation}
Afterwards, we pick the solution associated with $N^\star_t$, and apply the corresponding optimal input
\begin{align}\label{eq:cl_control}
    u^\star_{t|t}(x_t) = u^\star_t(x_t) = \bar{u}^\star_{t|t},
\end{align}
to system \eqref{eq:unc_system}, with $V^\mathrm{MPC}_{t \rightarrow t+N^\star_t}(x_t, N^\star_t) = J^\star(x_t)$.
We then resolve \eqref{eq:MPC_R_fin_trac} at $(t+1)$ for $N_{t+1} \in \{1,2,\dots,N\}$. 
\section{Feasibility and Stability}\label{sec:feas}
In this section we prove the feasibility and stability properties of the proposed robust MPC. 
\subsection{Feasibility}
\begin{theorem}\label{thm1}
Consider the closed-loop system \eqref{eq:unc_system} and \eqref{eq:cl_control}. Let problem \eqref{eq:MPC_R_fin_trac} be feasible at time step $t=0$ for some horizon length $N_t \in \{1,2,\dots,N\}$. Then problem~\eqref{eq:MPC_R_fin_trac} is feasible at all time steps $t\geq 1$ for some  horizon length $N_t \in \{1,2,\dots,N\}$, possibly time-varying.
\end{theorem}
\begin{proof}
See Appendix.  
\end{proof}
\subsection{Stability}
To prove the stability of the origin in closed-loop, we first introduce the following set of assumptions and definitions.
\begin{assumption}\label{assump:orig_in}
Denote the set of state and input constraints in \eqref{eq:FTOCP_constr} as $\mathcal{X}$ and $\mathcal{U}$, respectively. We assume that the convex, compact sets $\mathcal{X}, \mathcal{U}$ and $\mathbb{W}$ contain the origin in their interior. 
\end{assumption}
\begin{definition}[$N$-Step Robust Controllable Set]\label{def:RobustPre} Given a control policy $\pi(\cdot)$ and the closed-loop system $x_{t+1} = A x_t + B\pi(x_t) + w_t$ with $w_t \in \mathbb{W}$ for all $t\geq 0$, we recursively define the $N$-Step Robust Controllable set to the set $\mathcal{S}$ as 
\begin{equation*}
\begin{aligned}
    & \mathcal{C}_{t\rightarrow t+k+1}(\mathcal{S}) = \mathrm{Pre}(\mathcal{C}_{t\rightarrow t+k}(\mathcal{S}), A, B, \mathbb{W}, \pi(\cdot)) \cap \mathcal{X},\\[1ex]
    & \textnormal{with } \mathcal{C}_{t\rightarrow t}(\mathcal{S})=\mathcal{S},~\textnormal{for $k=\{0, 1, \dots, N-1 \}$},
\end{aligned}
\end{equation*}
where $\mathrm{Pre}(\mathcal{S},A,B,\mathbb{W}, \pi(\cdot))$ defines the set of states of the system $x_{t+1} = A x_t + B\pi(x_t) + w_t$, which evolve into the target set $\mathcal{S}$ in one time step for all $w_t \in \mathbb{W}$.
\end{definition}
An algorithm to compute an inner approximation of such a set is presented in \cite{rosolia2019robust, bujarbaruah2020robust}, which we call the approximate $N$-Step Robust Controllable Set.
\begin{definition}[ROA of the Robust MPC]\label{def:safes}
The ROA for the proposed robust MPC, denoted by $\mathcal{R}$, is defined as the 
union of the $N_t$-Step Robust Controllable Sets to the terminal set $\mathcal{X}_N$ under the policy \eqref{eq:cl_control}, for $N_t \in \{1,2,\dots,N\}$. 
\end{definition}
An inner approximation to the ROA, which we call the approximate ROA, can be obtained using the approximate $N$-Step Robust Controllable Sets.
\begin{assumption}\label{assump:stagecost} 
The matrices $P$ and $R$ in $\ell(x, u) = x^\top P x + u^\top R u$ are positive definite, i.e., $P\succ0$ and $R\succ0$.
\end{assumption}
\begin{assumption}\label{assump: termcost}
The matrix $P_N$ which defines the terminal cost in \eqref{eq:MPC_R_fin_trac} is chosen as
a matrix $P_N \succ 0$ satisfying 
\begin{align}\label{eq:lmi_ly}
 & x^\top \Big (-P_N + (P+K^\top R K) + \bar{A}_\mathrm{cl}^\top P_N \bar{A}_\mathrm{cl} \Big )x \leq 0
\end{align}
for all $x \in \mathcal{X}_N$, where $\bar{A}_\mathrm{cl} = \bar{A} + \bar{B}K$.
\end{assumption}
\begin{definition}[ISS Lyapunov Function \cite{lin1995various}]\label{iss_def}
Consider the closed-loop system given by
\begin{align}\label{eq:cl_loop_system}
    x_{t+1} = Ax_t + B{u}^\star_{t|t}(x_t) + w_t,~\forall t\geq 0.
\end{align}
Then the origin is called Input to State Stable (ISS), with a ROA $\mathcal{R} \subset \mathbb{R}^{d}$, if there exists class-$\mathcal{K}_\infty$ functions $\alpha_1(\cdot)$, $\alpha_2(\cdot)$, $\alpha_3(\cdot)$, a class-$\mathcal{K}$ function $\sigma(\cdot)$ and a function $V(\cdot): \mathbb{R}^d \mapsto \mathbb{R}_{\geq 0}$ continuous at the origin, such that, 
\begin{align*}
    & \alpha_1(\Vert x \Vert ) \leq V(x) \leq \alpha_2(\Vert x \Vert ),~\forall x \in \mathcal{R},\\
    & V(x_{t+1}) - V(x_t) \leq -\alpha_3(\Vert x_t \Vert) + \sigma(\Vert \tilde{w}_i \Vert_{\mathcal{L}_\infty}),
\end{align*}
where $\tilde{w}_i = \Delta^\mathrm{tr}_A x_i + \Delta^\mathrm{tr}_B u_i + w_i$ and $\Vert \tilde{w}_i \Vert_{\mathcal{L}_\infty} = \sup_{i \in \{0,\dots,t\}}\Vert \tilde{w}_i \Vert$. Function $V(\cdot)$ is called an ISS Lyapunov function for \eqref{eq:cl_loop_system}.
\end{definition}
\begin{theorem}\label{isstheorem}
Let Assumptions~\ref{assump:stable}-\ref{assump: termcost} hold and let $x_0 \in \mathcal{R}$. Then, the optimal cost of \eqref{eq:MPC_R_fin_trac} with \eqref{eq:nstar}, i.e., $J^{\star}(x_t)$ is an ISS Lyapunov function for the closed-loop system \eqref{eq:cl_loop_system}. This guarantees Input to State Stability of the origin of \eqref{eq:cl_loop_system}. 
\end{theorem}
\begin{proof}
See Appendix.
\end{proof}
\section{Numerical Simulations}\label{sec:numerics}
We choose $N=5$ and compute approximate solutions to the example problem given in \cite{bujarbaruah2020robust}.
The feedback gain $K$ satisfying Assumption~\ref{assump:stable} is chosen as $K= -[0.4866, 0.4374]$. The source codes are at \href{https://github.com/monimoyb/RMPC_SimpleTube}{\texttt{https://github.com /monimoyb/RMPC\_SimpleTube}}. 

\subsection{Comparison with \cite{langson2004robust}}
The tube cross section ($Z$) is chosen as the minimal robust positive invariant set \cite[Definition~3.4]{kouvaritakis2016model} for system \eqref{eq:unc_system} under a feedback $u = -[0.7701, 0.7936]x$, and the terminal set ($\mathcal{X}_f$) is chosen as $\mathcal{X}_N$ in \eqref{eq:term_set_DF}. See \cite{langson2004robust} for details on these quantities. We then choose a set of $N_\mathrm{init} = 100$ initial states $x_S$, created by a $10 \times 10$ uniformly spaced grid of the set of state constraints. From  each of these initial state samples we check the feasibility of the tube MPC problem in \cite[Section~5]{langson2004robust}. The code to solve the tube MPC is used from \cite{slscode}.  The convex hull of the feasible initial states (largest out of horizons $N\leq5$) inner approximates the ROA of the tube MPC. This is compared to the approximate ROA of our proposed robust MPC. The comparison is shown in Fig.~\ref{fig:init_sample_sls}. The approximate ROA from our approach is about 1.05x \emph{larger} in volume, but containing $98\%$ of that of the tube MPC.
\begin{figure}[h!]
\textcolor{yellow}{~~\!\textbf{$\blacksquare$}} Approx. ROA of Proposed Robust MPC \\ \textcolor{gray}{\textbf{$\blacksquare$}} Approx. ROA of Tube MPC in \cite{langson2004robust}\\[0.2cm]
	\centering
	\includegraphics[width=0.74\columnwidth]{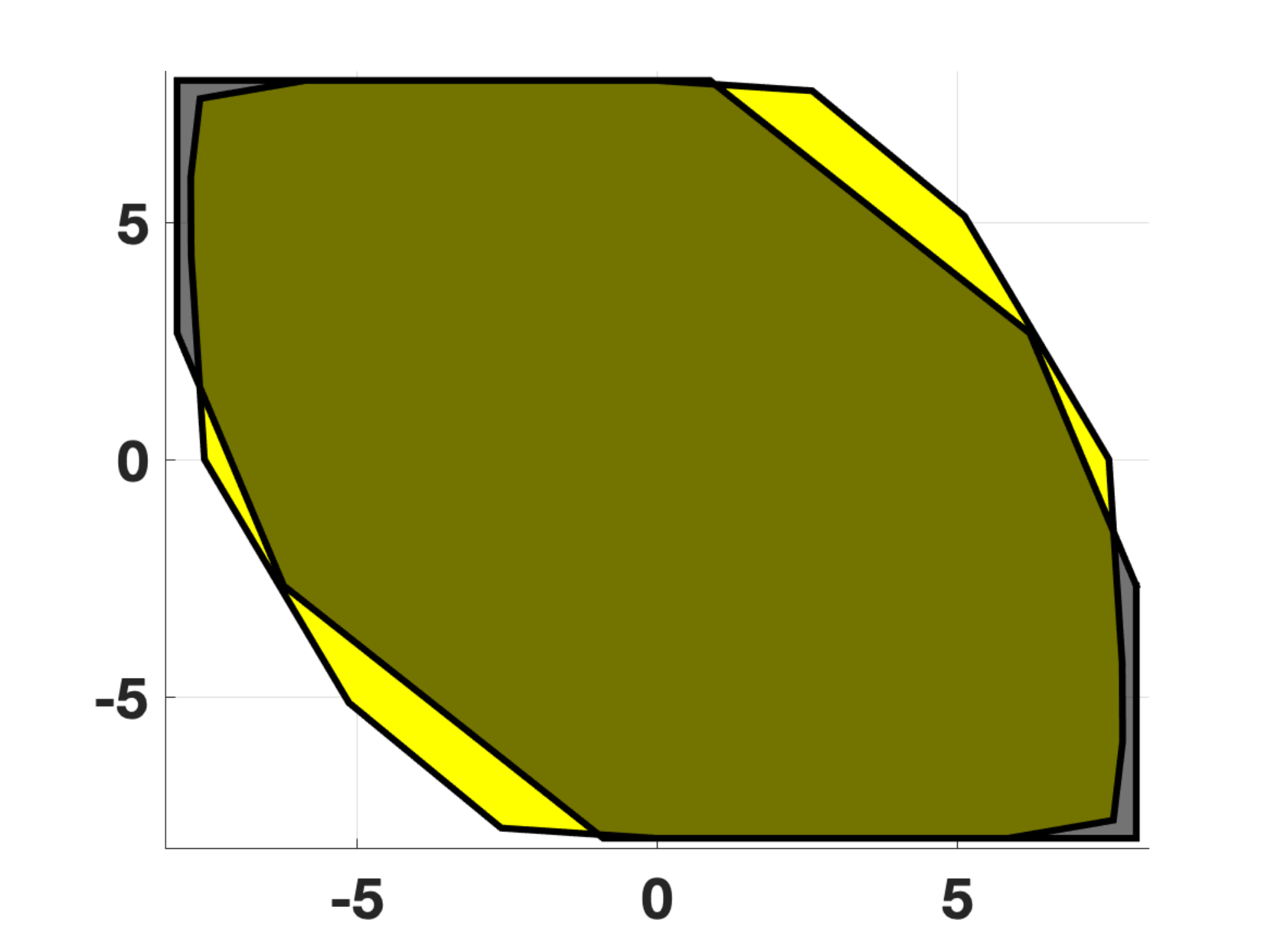}
	\caption{Comparison of the approximate ROA. 
	}
	\label{fig:init_sample_sls}    
\end{figure}
However, for any $N \leq 5$, the tube MPC needs higher computation times than for all $N_t \in \{1,2,\dots, N\}$ combined in our approach. This is shown in Table~\ref{tab:compTime}. 
\begin{table}[h]
\caption{Avg. online computation times [sec]. Values are obtained with a MacBook Pro 16inch, 2019, 2.3 GHz 8-Core Intel Core i9, 16 GB memory, using the Gurobi solver.}
\label{tab:compTime}
  \begin{center} 
  \begin{tabular}{|c|c|c|}
    \hline
   {\textbf{Horizon}} & \textbf{Proposed Robust MPC} & \textbf{Tube MPC in \cite{langson2004robust}}\\
    \hline
     $N_t=1$   &  0.0026  &  0.0062\\ \hline
     $N_t=2$   &  0.0023  &  0.0753\\ \hline
     $N_t=3$   &  0.0038  & 0.1612\\ \hline
     $N_t=4$   &  0.0056  &  0.2556\\ \hline
     $N_t=5$   &  0.0078  & 0.3384\\ \hline
  \end{tabular}
  \end{center}
  \vspace{-18pt}
\end{table}
\begin{remark} 
See \cite{bujarbaruah2020robust} on how to outperform the tube MPC both in conservatism and online computational complexity.
\end{remark}
\subsection{Roll-Out Alternative and Comparison with~\cite{dean2018safely}}\label{disc:II}
A computationally cheaper alternative can be obtained as follows: Once an optimization problem in \eqref{eq:MPC_R_fin_trac} at time step $t=0$ is feasible for some horizon length $N_0 = \bar{N}_0 \in \{1,2,\dots,N\}$, the corresponding optimal policy sequence: $\{u^\star_{0|0}, u^\star_{1|0}(\cdot), \dots, u^\star_{\bar{N}_0-1|0}(\cdot)\}$ can be used to obtain a 
safe open-loop policy for all time steps as:
\begin{align}\label{sarahpol}
  \Pi^\mathrm{safe}_\mathrm{ol}(x_t) = \begin{cases} u^\star_{t|0}(x_t), &\mbox{if } t \leq (\bar{N}_0-1), \\
    Kx_t, & \textnormal{otherwise}. \end{cases}
\end{align}
Policy \eqref{sarahpol} maintains the robust satisfaction of \eqref{eq:FTOCP_constr} for all time steps, without re-solving \eqref{eq:MPC_R_fin_trac}.
From  each of the previous 100 initial state samples, we now check the feasibility of the constrained LQR synthesis problem in \cite[Section~2.3]{dean2018safely}. We pick the FIR length (same as control horizon length) as $L=  15$,
with $\tau = 0.99$ and $\tau_\infty = 0.2$. See \cite[Problem~2.8]{dean2018safely} for details on these parameters. 
The comparison of the approximate $\bar{N}_0$-Step Robust Controllable Sets and the approximate region of attraction of the algorithm of \cite[Section~2.3]{dean2018safely} is shown in Fig.~\ref{fig:ol_samples}. 
The volumes of the approximate $\bar{N}_0$-Step Robust Controllable Sets are bigger than the approximate ROA of the controller in \cite[Section~2.3]{dean2018safely} for all $\bar{N}_0 \leq 5$, showing that the roll-out policy \eqref{sarahpol} yields up to approximately 12x lower conservatism.  
\begin{figure}[h]
\centering\textcolor{yellow}{\textbf{$\blacksquare$}} Approx.\ $\bar{N}_0$-Step Robust Controllable Set \\ \textcolor{blue}{\textbf{$\blacksquare$}} Approx.\ ROA of Controller in \cite{dean2018safely}\\[0.2cm]
\captionsetup[subfigure]{}
\centering
    \subfloat[$\bar{N}_0 = 2,3,4$ (comparable sets)]{%
        \includegraphics[width=0.5\columnwidth]{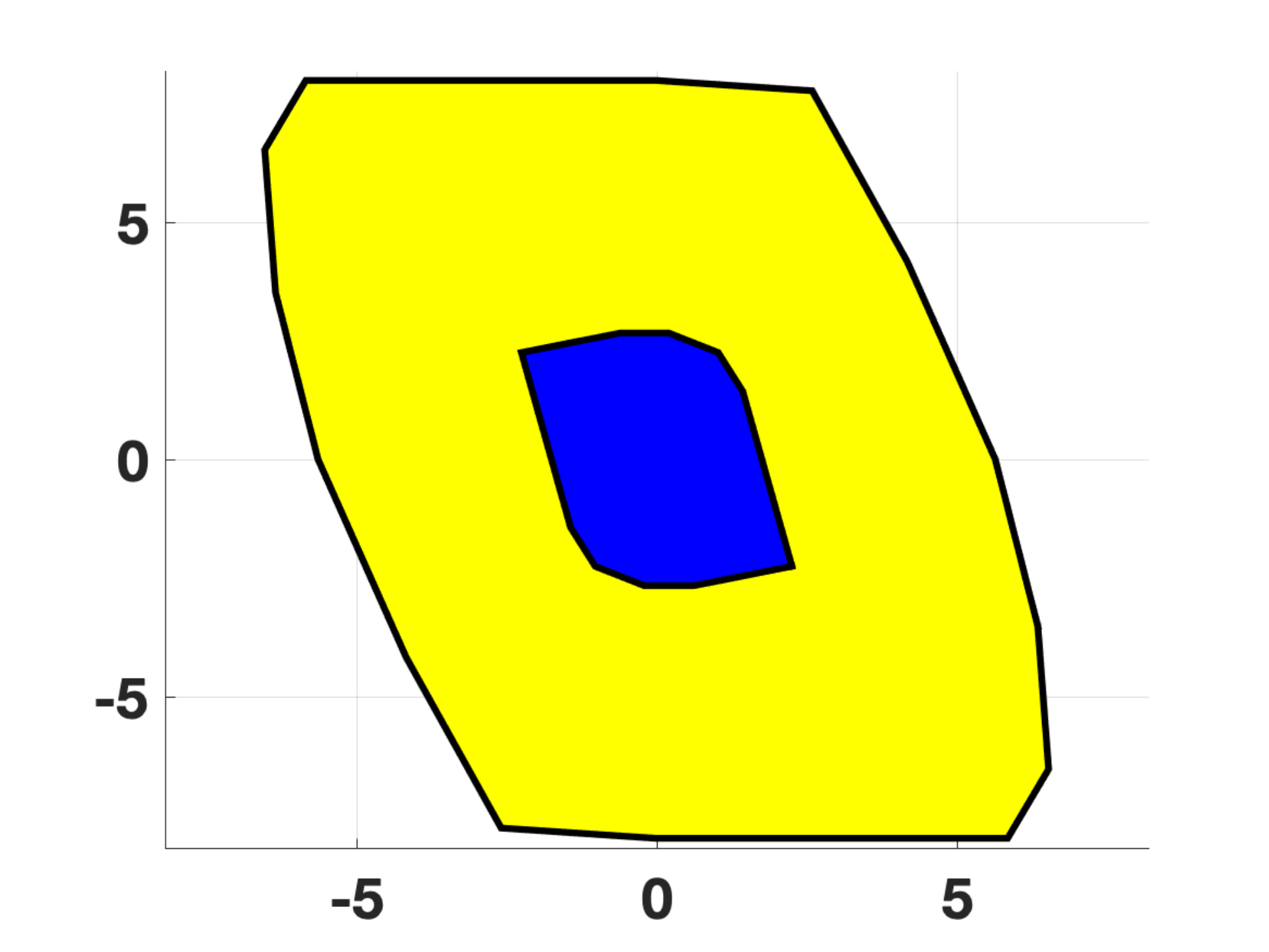}
        \label{fig:motionfig:dry}
    }
    \subfloat[$\bar{N}_0 = 5$]{%
        \includegraphics[width=0.5\columnwidth]{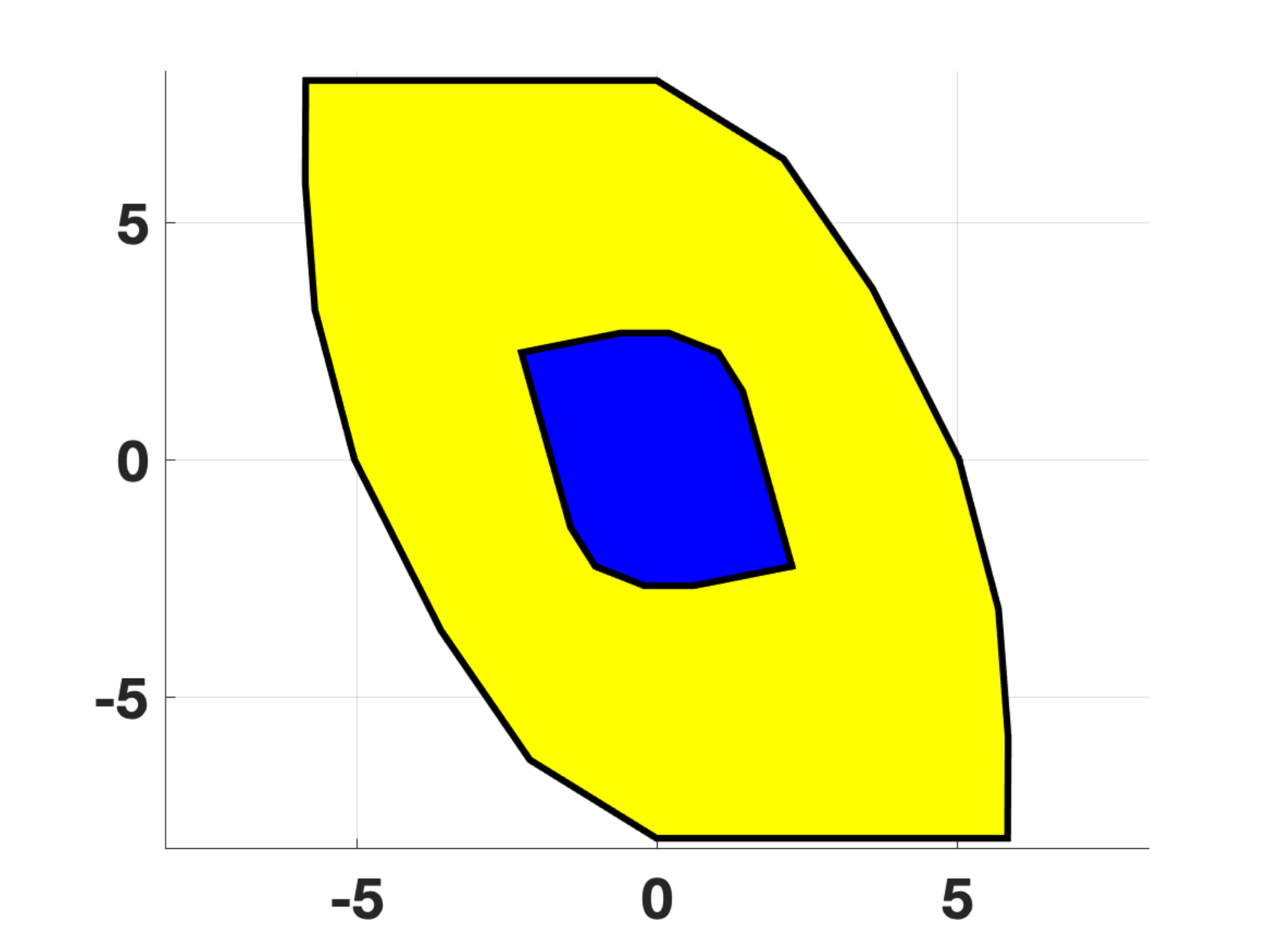}
        \label{fig:motionfig:wet}
    }
    \caption{A safe open-loop policy \eqref{sarahpol} is guaranteed to exist at all times with initial states in the yellow regions.}
    \label{fig:ol_samples}
    \vspace{-10pt}
\end{figure}
\section{Conclusions}
We proposed a computationally efficient approach to design a robust MPC for constrained uncertain linear systems. The uncertainty considered included both mismatch in the system dynamics matrices, and an additive disturbance. The designed MPC is recursively feasible and the origin of the closed-loop system is Input to State stable. 
With numerical simulations, we demonstrated that the proposed approach can be a simple and viable alternative to balance the trade-off between computational complexity and conservatism in robust MPC design under parametric model uncertainty. 
\renewcommand{\baselinestretch}{1}
\bibliographystyle{IEEEtran}
\bibliography{bibliography.bib}

\begin{thebibliography}{10}
\providecommand{\url}[1]{#1}
\csname url@samestyle\endcsname
\providecommand{\newblock}{\relax}
\providecommand{\bibinfo}[2]{#2}
\providecommand{\BIBentrySTDinterwordspacing}{\spaceskip=0pt\relax}
\providecommand{\BIBentryALTinterwordstretchfactor}{4}
\providecommand{\BIBentryALTinterwordspacing}{\spaceskip=\fontdimen2\font plus
\BIBentryALTinterwordstretchfactor\fontdimen3\font minus
  \fontdimen4\font\relax}
\providecommand{\BIBforeignlanguage}[2]{{%
\expandafter\ifx\csname l@#1\endcsname\relax
\typeout{** WARNING: IEEEtran.bst: No hyphenation pattern has been}%
\typeout{** loaded for the language `#1'. Using the pattern for}%
\typeout{** the default language instead.}%
\else
\language=\csname l@#1\endcsname
\fi
#2}}
\providecommand{\BIBdecl}{\relax}
\BIBdecl

\bibitem{mayne2000constrained}
D.~Q. Mayne, J.~B. Rawlings, C.~V. Rao, and P.~O. Scokaert, ``Constrained model
  predictive control: Stability and optimality,'' \emph{Automatica}, vol.~36,
  no.~6, pp. 789--814, 2000.

\bibitem{kouvaritakis2016model}
B.~Kouvaritakis and M.~Cannon, \emph{Model predictive control: Classical,
  robust and stochastic}.\hskip 1em plus 0.5em minus 0.4em\relax Springer,
  2016.

\bibitem{borrelli2017predictive}
F.~Borrelli, A.~Bemporad, and M.~Morari, \emph{Predictive control for linear
  and hybrid systems}.\hskip 1em plus 0.5em minus 0.4em\relax Cambridge
  University Press, 2017.

\bibitem{chisci2001systems}
L.~Chisci, J.~A. Rossiter, and G.~Zappa, ``Systems with persistent
  disturbances: predictive control with restricted constraints,''
  \emph{Automatica}, vol.~37, no.~7, pp. 1019--1028, 2001.

\bibitem{langson2004robust}
W.~Langson, I.~Chryssochoos, S.~Rakovi{\'c}, and D.~Q. Mayne, ``Robust model
  predictive control using tubes,'' \emph{Automatica}, vol.~40, no.~1, pp.
  125--133, 2004.

\bibitem{Goulart2006}
P.~J. Goulart, E.~C. Kerrigan, and J.~M. Maciejowski, ``Optimization over state
  feedback policies for robust control with constraints,'' \emph{Automatica},
  vol.~42, no.~4, pp. 523--533, 2006.

\bibitem{rakovic2012homothetic}
S.~V. Rakovi{\'c}, B.~Kouvaritakis, R.~Findeisen, and M.~Cannon, ``Homothetic
  tube model predictive control,'' \emph{Automatica}, vol.~48, no.~8, pp.
  1631--1638, 2012.

\bibitem{bujarbaruahAdapCDC18}
M.~Bujarbaruah, X.~Zhang, U.~Rosolia, and F.~Borrelli, ``Adaptive {MPC} for
  iterative tasks,'' \emph{2018 IEEE Conference on Decision and Control (CDC)},
  pp. 6322--6327, 2018.

\bibitem{rakovic2013homothetic}
S.~V. Rakovi{\'c} and Q.~Cheng, ``Homothetic tube {MPC} for constrained linear
  difference inclusions,'' in \emph{25th IEEE Chinese Control and Decision
  Conference (CCDC)}.\hskip 1em plus 0.5em minus 0.4em\relax IEEE, 2013, pp.
  754--761.

\bibitem{munoz2013recursively}
D.~Mu{\~n}oz-Carpintero, M.~Cannon, and B.~Kouvaritakis, ``Recursively feasible
  robust {MPC} for linear systems with additive and multiplicative uncertainty
  using optimized polytopic dynamics,'' in \emph{Conference on Decision and
  Control}.\hskip 1em plus 0.5em minus 0.4em\relax IEEE, 2013, pp. 1101--1106.

\bibitem{slsmpc}
S.~Chen, H.~Wang, M.~Morari, V.~M. Preciado, and N.~Matni, ``Robust closed-loop
  model predictive control via system level synthesis,'' in \emph{Conference on
  Decision and Control}.\hskip 1em plus 0.5em minus 0.4em\relax IEEE, 2020, pp.
  2152--2159.

\bibitem{dean2018safely}
S.~Dean, S.~Tu, N.~Matni, and B.~Recht, ``Safely learning to control the
  constrained linear quadratic regulator,'' \emph{arXiv preprint
  arXiv:1809.10121}, 2018.

\bibitem{rawlings2009model}
J.~B. Rawlings and D.~Q. Mayne, \emph{Model predictive control: Theory and
  design}.\hskip 1em plus 0.5em minus 0.4em\relax Nob Hill Pub., 2009.

\bibitem{lofberg2003minimax}
J.~L{\"o}fberg, \emph{Minimax approaches to robust model predictive
  control}.\hskip 1em plus 0.5em minus 0.4em\relax Link{\"o}ping University
  Electronic Press, 2003, vol. 812.

\bibitem{rosolia2019robust}
U.~Rosolia, X.~Zhang, and F.~Borrelli, ``Robust learning model predictive
  control for linear systems performing iterative tasks,'' \emph{arXiv preprint
  arXiv:1911.09234}, 2019.

\bibitem{bujarbaruah2020robust}
M.~Bujarbaruah, U.~Rosolia, Y.~R. St{\"u}rz, X.~Zhang, and F.~Borrelli,
  ``Robust {MPC} for linear systems with parametric and additive uncertainty: A
  novel constraint tightening approach,'' \emph{arXiv preprint
  arXiv:2007.00930}, 2020.

\bibitem{lin1995various}
Y.~Lin, E.~Sontag, and Y.~Wang, ``Various results concerning set input-to-state
  stability,'' in \emph{Conference on Decision and Control}, vol.~2.\hskip 1em
  plus 0.5em minus 0.4em\relax IEEE, 1995, pp. 1330--1335.

\bibitem{slscode}
S.~Chen and H.~Wang, ``Robust-{MPC}-{SLS} repository,'' \emph{URL
  https://github.com/unstable-zeros/robust-mpc-sls}, 2020.

\end{thebibliography}

\section*{Appendix}
\subsection{Matrices in \eqref{ng1}}
As in~\cite{Goulart2006}, matrices $\bar{\mathbf{A}}$, $\mathbf{C}$ and $\mathbf{G}$ for a horizon $\bar{N}$ are: $\mathbf{G} = I_{d\bar{N}} + \sum_{k=1}^{\bar{N}-1}\mathcal{L}_{\bar N}^k \otimes \bar{A}^k ,\bar{\mathbf{A}} = \mathrm{diag}(\bar{A}, \bar{A}^2, \dots, \bar{A}^{\bar{N}-1})$, and  $\mathbf{C} = \mathbf{G} \cdot (I_{\bar{N}} \otimes \bar{B})$, with $\mathcal{L}$ being the lower shift matrix.
\subsection{Proof of Theorem~\ref{thm1}}
Assume that at time step $t$ problem \eqref{eq:MPC_R_fin_trac} is feasible, and let $N_t^\star$ be the optimal horizon. We then consider: 

\noindent \textbf{Case 1: ($N^\star_t = 1$)}
Consider the robust state constraints \eqref{n1}:
\begin{align}\label{eq:state_conN1}
    \max_{\substack{{w}_t \in \mathbb{W}\\ \Delta_A \in \mathcal{P}_A, \Delta_B \in \mathcal{P}_B}}  \!\!\!\! \!\!\!\! \!\!H_N^x((\bar{{A}}+{\Delta}_A) {x}_t + ({\bar{B}} + {\Delta}_B) \bar{\mathbf{u}}^{(1)}_t + {w}_t) \leq h_N^x.
\end{align}
We find $h_N^x$ where the max is attained by using duality. 
Let us denote the corresponding optimal input policy by 
\begin{align}\label{lem_proof_pol}
u^\star_{t|t}(x_t) = \bar{u}^\star_{t|t}.
\end{align}
Now, let policy \eqref{lem_proof_pol} be applied to \eqref{eq:unc_system} in closed-loop, so that the system reaches the terminal set $\mathcal{X}_N$. Consider solving \eqref{eq:state_conN1} at this step with a horizon length of $N_{t+1}=1$. As \eqref{ng1} uses the same representation of the uncertainty as done in Section~\ref{ssec:term_set}, a candidate policy at time step $(t+1)$ is 
\begin{align}\label{eq:can_tp1}
    u_{t+1|t+1}(x_{t+1}) = Kx_{t+1},
\end{align}
which is a feasible solution to \eqref{eq:MPC_R_fin_trac} under constraint \eqref{eq:state_conN1}. 

\noindent \textbf{Case 2: ($N^\star_t \geq 2$)}
Let us denote the sequence of optimal input policies from $t$ as $\{u^\star_{t|t},u^\star_{t+1|t}(\cdot),\cdots,u^\star_{t+N^\star_t-1|t}(\cdot)\}$. Consider a candidate policy sequence at the next time instant:
\begin{align}\label{eq:feas_seq_next_DF_sto}
    U_{t+1}(\cdot) = \{u^\star_{t+1|t}(\cdot),\dots,u^\star_{t+N^\star_t-1|t}(\cdot)\}.
\end{align}
Now using standard MPC shifting arguments \cite{chisci2001systems, langson2004robust, Goulart2006}, sequence \eqref{eq:feas_seq_next_DF_sto} is a feasible policy sequence at time step $(t+1)$ for problem \eqref{eq:MPC_R_fin_trac}, with horizon length $N_{t+1} = N^\star_t -1$. 
\subsection{Proof of Theorem~\ref{isstheorem}}
From Assumption~\ref{assump:stagecost} we know that, $\alpha_1(\Vert x_t \Vert_2) \leq \ell(x,0) \leq J^{\star}(x_t)$ for some $\alpha_1(\cdot) \in \mathcal{K}_\infty$ and for all $x \in \mathcal{R}$. Moreover, since \eqref{eq:MPC_R_fin_trac} can be reformulated into a parametric QP for each horizon length $N_t$, constraint set \eqref{eq:FTOCP_constr} is compact, and $J^{\star}(0) = 0$, from \cite[Theorem~23]{Goulart2006}, we know $J^{\star}(x_t) \leq \alpha_2(\Vert x_t \Vert_2)$ for some $\alpha_2(\cdot) \in \mathcal{K}_\infty$ and for all $x_t \in \mathcal{R}$. We complete the proof by considering the same two cases : \vspace{3pt}\\
\noindent \textbf{Case 1: ($N^\star_t = 1$)}
Consider the case of $N^\star_t = 1$. The optimal nominal cost at time step $t$ is written as 
\begin{subequations}
\begin{align}
J^{\star}(&x_{t}) \nonumber 
= \ell(\bar{x}^\star_{t|t},\bar{u}^\star_{t|t}) + (\bar{x}^\star_{t+1|t})^\top P_N \bar{x}^\star_{t+1|t} \nonumber \\
&  \geq \ell(\bar{x}^\star_{t|t},\bar{u}^\star_{t|t}) + \ell(\bar{x}^\star_{t+1|t},K\bar{x}^\star_{t+1|t}) + \nonumber \\
&~ + ((\bar{A}+\bar{B}K)\bar{x}^\star_{t+1|t})^\top P_N ((\bar{A}+\bar{B}K)\bar{x}^\star_{t+1|t})\label{assump:4used}\\
& = \ell(\bar{x}^\star_{t|t},\bar{u}^\star_{t|t}) + q(\bar{x}^\star_{t+1|t}), \label{eq:iss12}
\end{align}
\end{subequations}
where in \eqref{assump:4used} we have used Assumption~\ref{assump: termcost}, and at time step $(t+1)$ the feasible input $\bar{u}_{t+1|t} = K \bar{x}^\star_{t+1|t}$ as discussed in \eqref{eq:can_tp1}. As \eqref{eq:can_tp1} is a feasible policy at time step $(t+1)$ with horizon length $N_{t+1} = 1$, the optimal cost of the MPC problem for any horizon length $N^\star_{t+1}=\{1,2,\dots,N\}$ can be bounded from above as:
\begin{align}\label{eq:iss22}
 J^{\star}(x_{t+1}) & \leq \ell(\bar{x}_{t+1|t+1},\bar{u}_{t+1|t}(\bar{x}_{t+1|t+1})) + Q(\bar{x}_{t+2|t+1}) \nonumber \\
&= q(\bar{x}_{t+1|t+1}),
\end{align}
with $\bar{x}_{t+1|t+1}= \bar{x}^\star_{t+1|t} + \tilde{w}_t,~\textnormal{with } \tilde{w}_t = \Delta^\mathrm{tr}_A x_t + \Delta^\mathrm{tr}_B \bar{u}^{\star}_{t|t} + w_t$. Combining \eqref{eq:iss12}--\eqref{eq:iss22} we obtain:
\begin{equation}\label{iss_proof_12}
\begin{aligned}
    & J^{\star}(x_{t+1}) - J^{\star}(x_t)  \\
    & \leq q(\bar{x}^\star_{t+1|t} + \tilde{w}_t) - \ell(\bar{x}^\star_{t|t},\bar{u}^\star_{t|t}) - q(\bar{x}^\star_{t+1|t})\\
    & \leq -\alpha_3(\Vert x_t \Vert_2 ) + L_q \Vert \tilde{w}_i \Vert_{\mathcal{L}_\infty},
\end{aligned}
\end{equation}
where $q(\cdot)$ is $L_q$-Lipschitz as it is a sum of quadratics in $\mathcal{X}$. \vspace{3pt}\\
\noindent \textbf{Case 2: ($N^\star_t \geq 2$)}
Now consider 
\begin{align}\label{eq:iss1}
J^{\star}(x_{t}) & = \sum \limits_{k=t}^{t+N^\star_t-1} \ell(\bar{x}^\star_{k|t},\bar{u}^\star_{k|t}) + Q(\bar{x}^\star_{t+N^\star_t|t}) \nonumber \\
& = \ell(\bar{x}^\star_{t|t},\bar{u}^\star_{t|t}) + q(\bar{x}^\star_{t+1|t}),
\end{align}
where $\{\bar{x}^\star_{t|t},\bar{x}^\star_{t+1|t}, \dots,\bar{x}^\star_{t+N^\star_t|t}\}$ is the optimal predicted nominal trajectory under the optimal nominal input sequence $\{\bar{u}^\star_{t|t}, \bar{u}^\star_{t+1|t}, \dots, \bar{u}^\star_{t+N^\star_t-1|t}\}$, where $\bar{u}^\star_{k|t} = u^\star_{k|t}(\bar{x}^\star_{k|t})$ for all $k \in \{t,t+1,\dots, t+(N^\star_t-1)\}$. The quantity $q(\bar{x}^\star_{t+1|t})$ provides the total nominal cost from time step $(t+1)$ to $(t+N^\star_t)$ under the following optimal control policy
\begin{equation}\label{eq:feas_pol_iss_ugo}
    \{{u}^\star_{t+1|t}(\cdot), \ldots, u^\star_{t+N_t^*-1|t}(\cdot)\}.
\end{equation}
We know that \eqref{eq:feas_seq_next_DF_sto} is a feasible policy sequence for \eqref{eq:MPC_R_fin_trac} at time step $(t+1)$ with horizon length $N_{t+1} = (N^\star_t -1)$. After $\bar{x}_{t+1} = x_{t+1}$ is obtained with closed-loop system evolution \eqref{eq:cl_loop_system}, with this feasible policy sequence \eqref{eq:feas_pol_iss_ugo}, the optimal nominal cost of \eqref{eq:MPC_R_fin_trac} at time step $(t+1)$ for any $N^\star_{t+1} \in \{1,2,\dots,N\}$ can be bounded as:
\begin{align}\label{iss2}
    & J^{\star}(x_{t+1}) \leq 
    \sum \limits_{k=t+1}^{t+N^\star_t-1} \ell(\bar{x}_{k|t+1},{u}^\star_{k|t}(\bar{x}_{k|t+1})) + Q(\bar{x}_{t+N^\star_t|t+1})\nonumber \\
    &~~~~~~~~~~~ = q(\bar{x}_{t+1|t+1}),
\end{align}
where we have used the feasible nominal trajectory obtained with the policy  \eqref{eq:feas_pol_iss_ugo}, given as
\begin{align*}
\bar{x}_{k|t+1} = \bar{A}^{k-t-1}(\bar{A}x_t & + \bar{B}{u}_{t|t}^\star(x_t) + \tilde{w}_t) +  \\ 
&  + \sum \limits_{i=t+1}^{k-1} \bar{A}^{k-1-i}\bar{B} u^\star_{i|t}(\bar{x}_{k|t+1}),
\end{align*}
for $k = \{t+2,t+3,\dots,t+N^\star_t\}$, 
Moreover, we know that
\begin{align}\label{iss3}
    \bar{x}_{t+1|t+1}  = \bar{x}^\star_{t+1|t} + \tilde{w}_t,
\end{align}
with $\tilde{w}_t = \Delta^\mathrm{tr}_A x_t + \Delta^\mathrm{tr}_B \bar{u}^{\star}_{t|t} + w_t$. Combining \eqref{eq:iss1}--\eqref{iss3}:
\begin{equation}\label{iss_proof_1}
\begin{aligned}
    & J^{\star}(x_{t+1}) - J^{\star}(x_t)  \\
    & = q(\bar{x}^\star_{t+1|t} + \tilde{w}_t) - \ell(\bar{x}^\star_{t|t},\bar{u}^\star_{t|t}) - q(\bar{x}^\star_{t+1|t})\\
    & \leq - \ell(\bar{x}^\star_{t|t},\bar{u}^\star_{t|t}) + L_q \Vert \tilde{w}_t \Vert\leq - \ell(\bar{x}^\star_{t|t},0) + L_q \Vert \tilde{w}_t \Vert\\
    & \leq -\alpha_3(\Vert x_t \Vert_2 ) + L_q \Vert \tilde{w}_i \Vert_{\mathcal{L}_\infty}.
\end{aligned}
\end{equation}
Combining \eqref{iss_proof_12} and \eqref{iss_proof_1}, the origin of \eqref{eq:cl_loop_system} is ISS according to Definition~\ref{iss_def}, as the optimal cost function $J^\star(\cdot)$ is an ISS Lyapunov function. This completes the proof. 

\section*{Acknowledgements}
We thank Sarah Dean for constrained LQR source codes. Sponsors: ONR-N00014-18-1-2833, NSF-1931853, Marie Sk\l{}odowska-Curie grant 846421, and Ford motor company.
\end{document}